\documentclass[a4paper,UKenglish]{lipics-v2016}
\usepackage{tikz}
  \usetikzlibrary{arrows,automata,shapes}
  \usetikzlibrary{decorations.pathmorphing}
  \tikzset{snake arrow/.style= {->, decorate, decoration={snake,amplitude=.4mm,segment length=1.7mm,post length=2mm}}}
\usepackage{microtype}
\usepackage{xspace}

\newcommand{\fun}[1]{\ensuremath{\textsl{#1}}}
\newcommand{\eps}{\varepsilon}
\DeclareMathOperator{\A}{\mathcal A}
\DeclareMathOperator{\B}{\mathcal B}
\DeclareMathOperator{\D}{\mathcal D}
\DeclareMathOperator{\M}{\mathcal M}
\newcommand{\R}{\mathcal{R}}
\newcommand{\J}{\mathcal{J}}

\newcommand{\tuple}[1]{\langle{#1}\rangle}
\newcommand{\blank}{\llcorner\!\lrcorner}
\newcommand{\Deltaplus}{\Delta_{\#\$}}
\newcommand{\enc}{\fun{enc}}

\newcommand{\complclass}[1]{{\sc #1}\xspace}
\newcommand{\ACzero}{\complclass{AC${}^0$}}
\newcommand{\LogSpace}{\complclass{L}}
\newcommand{\NL}{\complclass{NL}}
\newcommand{\coNP}{\complclass{coNP}}
\newcommand{\PSpace}{\complclass{PSpace}}
\newcommand{\complete}{comp.}

\theoremstyle{plain}

\makeatletter
\newcommand{\newreptheorem}[2]{\newtheorem*{rep@#1}{\rep@title}\newenvironment{rep#1}[1]{\def\rep@title{#2 \ref*{##1}}\begin{rep@#1}}{\end{rep@#1}}}
\makeatother
\newreptheorem{theorem}{Theorem}
\newreptheorem{lemma}{Lemma}
\newreptheorem{corollary}{Corollary}

\allowdisplaybreaks

\title{Universality of Confluent, Self-Loop Deterministic Partially Ordered NFAs is Hard%
\footnote{This work was supported by the German Research Foundation (DFG) within the Collaborative Research Center SFB 912 (HAEC) and in Emmy Noether grant KR~4381/1-1 (DIAMOND)}}

\titlerunning{Universality of Confluent, Self-Loop Deterministic poNFAs is Hard}

\author[1]{Tom\'{a}\v{s} Masopust}

\author[2]{Markus Kr\"otzsch}

\affil[1]{Institute of Theoretical Computer Science and Center of Advancing Electronics Dresden (cfaed), TU Dresden, Germany,
  \texttt{tomas.masopust@tu-dresden.de}
}

\affil[2]{Institute of Theoretical Computer Science and Center of Advancing Electronics Dresden (cfaed), TU Dresden, Germany,
  \texttt{markus.kroetzsch@tu-dresden.de}
}

\authorrunning{T. Masopust, M. Kr\"otzsch} 

\Copyright{Tom\'{a}\v{s} Masopust, Markus Kr\"otzsch} 



\EventEditors{}
\EventNoEds{2}
\EventLongTitle{}
\EventShortTitle{ArXiv Report}
\EventAcronym{}
\EventYear{}
\EventDate{}
\EventLocation{}
\EventLogo{}
\SeriesVolume{}
\ArticleNo{}

\begin{document}

\maketitle

\begin{abstract}
  An automaton is partially ordered if the only cycles in its transition diagram are self-loops. The expressivity of partially ordered NFAs (poNFAs) can be characterized by the Straubing-Th\'e\-ri\-en hierarchy. Level~$3/2$ is recognized by poNFAs, level~1 by confluent, self-loop deterministic poNFAs as well as by confluent poDFAs, and level~$1/2$ by saturated poNFAs. We study the universality problem for confluent, self-loop deterministic poNFAs. It asks whether an automaton accepts all words over its alphabet. Universality for both NFAs and poNFAs is a \PSpace-complete problem. For confluent, self-loop deterministic poNFAs, the complexity drops to \coNP-complete if the alphabet is fixed but is open if the alphabet may grow. We solve this problem by showing that it is \PSpace-complete if the alphabet may grow polynomially. Consequently, our result provides a lower-bound complexity for some other problems, including inclusion, equivalence, and $k$-piecewise testability. Since universality for saturated poNFAs is easy, confluent, self-loop deterministic poNFAs are the simplest and natural kind of NFAs characterizing a well-known class of languages, for which deciding universality is as difficult as for general NFAs. 
\end{abstract}

\section{Introduction}
  McNaughton and Papert~\cite{McNaughton1971} showed that {\em first-order logic\/} describes {\em star-free languages}, a class of regular languages whose syntactic monoid is {\em aperiodic\/}~\cite{Schutzenberger65a}. Restricting the number of alternations of quantifiers in formulae in the prenex normal form results in a quantifier alternation hierarchy. The choice of predicates provides several hierarchies. The well-known and closely related 
  are the \emph{Straubing-Th\'erien (ST) hierarchy}~\cite{Straubing81,Therien81} and the \emph{dot-depth hierarchy}~\cite{BrzozowskiK78,CohenB71,Straubing85}.
  
  We are interested in automata characterizations of the levels of the ST hierarchy, alternatively defined as follows. For an alphabet $\Sigma$, $\mathscr{L}(0)=\{\emptyset, \Sigma^*\}$ and, for integers $n\geq 0$, 
    level $\mathscr{L}(n+1/2)$ consists of all finite unions of languages $L_0 a_1 L_1 a_2 \ldots a_k L_k$ with $k\geq 0$, $L_0,\ldots, L_k\in\mathscr{L}(n)$, and $a_1,\ldots,a_k\in\Sigma$, and
    level $\mathscr{L}(n+1)$ consists of all finite Boolean combinations of languages from level $\mathscr{L}(n+{1}/{2})$.
  The hierarchy does not collapse on any level~\cite{BrzozowskiK78}.
  For some levels, an algebraic or automata characterization is known. This characterization is particularly interesting in decision and complexity questions, such as the membership of a language in a specific level of the hierarchy. Despite a recent progress~\cite{AlmeidaBKK15,Place15,PlaceZ15}, deciding whether a language belongs to level $k$ of the ST hierarchy is still open for $k>{7}/{2}$. 

  The most studied level of the ST hierarchy is level 1, known as {\em piecewise testable languages\/} introduced by Simon~\cite{Simon1972}. Simon showed that piecewise testable languages are those regular languages whose syntactic monoid is $\J$-trivial and that they are recognized by {\em confluent, partially ordered DFAs}.
  An automaton is {\em partially ordered\/} if its transition relation induces a partial order on states -- the only cycles are self-loops -- and it is {\em confluent\/} if for any state $q$ and any two of its successors $s$ and $t$ accessible from $q$ by transitions labeled by $a$ and $b$, respectively, there is a word $w\in\{a,b\}^*$ such that a common state is reachable from both $s$ and $t$ under $w$, cf. Figure~\ref{fig_bad_pattern} (left) for an illustration.
  
  \begin{figure}
    \centering
    \begin{tikzpicture}[baseline,->,>=stealth,auto,shorten >=1pt,node distance=1.5cm,
      state/.style={circle,minimum size=5mm,very thin,draw=black,initial text=}]
      \node[state]  (1) {$q$};
      \node         (0) [right of=1]  {};
      \node[state]  (2) [above of=0,node distance=.5cm]  {$s$};
      \node[state]  (3) [below of=0,node distance=.5cm]  {$t$};
      \node[state]  (4) [right of=0,node distance=3cm] {};
      \path
        (1) edge[bend left=15]  node {$a$} (2)
        (1) edge[bend right=15] node[below] {$b$} (3)
        (2) edge[snake arrow,bend left=10] node[above] {$w\in\{a,b\}^*$} (4)
        (3) edge[snake arrow,bend right=10] node[below] {$w\in\{a,b\}^*$} (4)
        ;
    \end{tikzpicture}
    \hspace{3cm}
    \begin{tikzpicture}[baseline,->,>=stealth,auto,shorten >=1pt,node distance=2cm,
      state/.style={circle,minimum size=5mm,very thin,draw=black,initial text=}]
      \node[state]  (a) {};
      \node[state]  (aa) [right of=a]  {};
      \path
        (a) edge[loop above] node {$a$} (a)
        (a) edge node {$a$} (aa)
        ;
    \end{tikzpicture}
    \caption{Confluence (left) and the forbidden pattern of self-loop deterministic poNFAs (right)}
    \label{fig_bad_pattern}
  \end{figure}
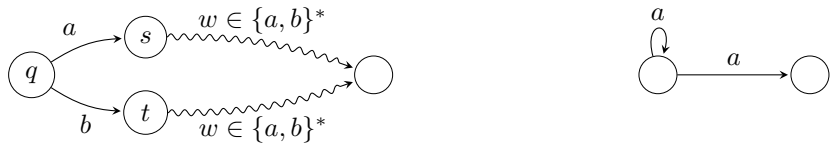

  Omitting confluence results in {\em partially ordered DFAs\/} (poDFAs) studied by Brzozowski and Fich~\cite{BrzozowskiF80}. They showed that poDFAs characterize {\em $\R$-trivial languages}, a class of languages strictly between level 1 and level ${3}/{2}$ of the ST hierarchy. Lifting the notion from DFAs to NFAs, Schwentick, Th\'erien and Vollmer~\cite{SchwentickTV01} showed that partially ordered NFAs (poNFAs) characterize level ${3}/{2}$ of the ST hierarchy. Hence poNFAs are more powerful than poDFAs. Languages of level $3/2$ are also known as {\em Alphabetical Pattern Constraints}~\cite{BMT2001}, which are regular languages effectively closed under permutation rewriting. 

  In our recent work, we showed that the increased expressivity of poNFAs is caused by self-loop transitions involved in nondeterminism. Consequently, $\R$-trivial languages are characterized by self-loop deterministic poNFAs~\cite{mfcs16:mktmmt}. A poNFA is {\em self-loop deterministic\/} if it does not contain the pattern of Figure~\ref{fig_bad_pattern} (right). Our study further revealed that complete, confluent and self-loop deterministic poNFAs characterize piecewise testable languages~\cite{ptnfas,dlt15}. An NFA is {\em complete\/} if a transition under every letter is defined in every state. Complete, confluent and self-loop deterministic poNFAs are thus a natural extension of confluent poDFAs to nondeterministic automata.
  
  In this paper, we study the {\em universality\/} problem of complete, confluent and self-loop deterministic poNFAs. The problem asks whether a given automaton accepts all words over its alphabet. 
  The study of universality (and its dual, emptiness) has a long tradition in formal languages 
  with many applications across computer science, e.g., in knowledge representation and database theory~\cite{BarceloLR:jacm14,CalvaneseGLV03:rpqreasoning,SMKR:elcq14}. The problem is \PSpace-complete for NFAs~\cite{MeyerS72}. Recent studies investigate it for specific types of regular languages, such as prefixes or factors~\cite{RampersadSX12}.

  In spite of a rather low expressivity of poNFAs, the universality problem for poNFAs has the same worst-case complexity as for general NFAs, even if restricted to binary alphabets~\cite{mfcs16:mktmmt}. This might be because poNFAs possess quite a powerful nondeterminism. The pattern of Figure~\ref{fig_bad_pattern} (right), which may occur in poNFAs, admits an unbounded number of nondeterministic steps -- the poNFA either stays in the same state or moves to another one. Forbidding the pattern results in self-loop deterministic poNFAs where the number of nondeterministic steps is bounded by the number of states. This restriction affects the complexity of universality. Deciding universality of self-loop deterministic poNFAs is \coNP-complete if the alphabet is fixed but remains \PSpace-complete if the alphabet may grow polynomially~\cite{mfcs16:mktmmt}. The growth of the alphabet thus compensates for the restriction on the number of nondeterministic steps. The reduced complexity is also preserved by complete, confluent and self-loop deterministic poNFAs if the alphabet is fixed~\cite{ptnfas} but is open if the alphabet may grow.

  We solve this problem by showing that deciding universality for complete, confluent and self-loop deterministic poNFAs is \PSpace-complete if the alphabet may grow polynomially, which strengthens our previous result for self-loop deterministic poNFAs~\cite{mfcs16:mktmmt}. 
  
  Consequently, the $k$-piecewise testability problem for complete, confluent and self-loop deterministic poNFAs, which was open~\cite{ptnfas}, is \PSpace-complete. The problem asks whether a given language is a finite boolean combination of languages of the form $\Sigma^* a_1 \Sigma^* \cdots \Sigma^* a_n \Sigma^*$, where $a_i\in \Sigma$ and $0\le n\le k$. It is of interest in XML databases and separability~\cite{HofmanM15,MartensNNS15}. The result follows from the fact that $0$-piecewise testability coincides with universality, and from the results in Masopust~\cite{ptnfas}. The problem is \coNP-complete for DFAs~\cite{KKP}.
  
  Another consequence is the worst-case complexity of the inclusion and equivalence problems for restricted NFAs, problems that are of interest, e.g., in optimization. The problems ask, given languages $K$ and $L$, whether $K\subseteq L$, resp. $K=L$. Universality can be expressed as the inclusion $\Sigma^* \subseteq L$ or as the equivalence $\Sigma^* = L$. Although equivalence means two inclusions, complexities of these two problems may differ significantly -- inclusion is undecidable for deterministic context-free languages~\cite{Friedman76} while equivalence is decidable~\cite{Senizergues97}. The complexity of universality provides a lower bound on the complexity of both. Deciding inclusion or equivalence of two languages given as complete, confluent and self-loop deterministic poNFAs is thus \PSpace-complete in general, \coNP-complete if the alphabet is fixed, and \NL-complete if the alphabet is unary, see Masopust~\cite{ptnfas} for the lower bound of the second result and Kr\"otzsch et al.~\cite{mfcs16:mktmmt_full} for the upper bound of the last two results.

  To complete the picture, H\'eam~\cite{Heam02} characterized level~${1}/{2}$ of the ST hierarchy as languages recognized by saturated poNFAs (spoNFA), also called {\em shuffle ideals\/} or {\em upward closures\/}. A poNFA is {\em saturated\/} if it has a self-loop under every letter in every state. Deciding universality for spoNFAs is simple -- it means to find a state that is both initial and accepting. 
  Therefore, complete, confluent and self-loop deterministic poNFAs are the simplest and natural kind of NFAs recognizing a well-known class of languages for which the universality problem is as difficult as for general NFAs.

  The following notation has been used for poNFAs and we adopt it in the sequel:
  \begin{itemize}
    \item self-loop deterministic poNFAs are denoted by {\em restricted poNFAs} or {\em rpoNFAs}~\cite{mfcs16:mktmmt}, and
    \item complete, confluent and self-loop deterministic poNFAs are denoted by {\em ptNFAs}~\cite{ptnfas}, where pt stands for piecewise testable.
  \end{itemize}

  \begin{table}
    \begin{center}
      \caption{Complexity of deciding universality for poNFAs and special classes thereof; ST stands for the corresponding level of the ST hierarchy; $\Sigma$ denotes the input alphabet}
      \label{table_results}
      \begin{tabular}{r|c|r@{\,}l|r@{\,}l|r@{\,}l}
          & ST
          & \multicolumn{2}{@{}c|}{$|\Sigma|=1$} 
          & \multicolumn{2}{@{}c|}{$|\Sigma|=k\ge 2$} 
          & \multicolumn{2}{@{}c}{$\Sigma$ is growing}\\
        \hline
        DFA       &
                  & L-\complete & \cite{Jones75}
                  & \NL-\complete & \cite{Jones75}
                  & \NL-\complete & \cite{Jones75}\\
        spoNFA    & $\frac{1}{2}$
                  & \ACzero & (Thm.~\ref{thmSaturated})
                  & \ACzero & (Thm.~\ref{thmSaturated})
                  & \ACzero & (Thm.~\ref{thmSaturated})\\
        ptNFA     & $1$
                  & \NL-\complete & (Thm.~\ref{thmMainNL})
                  & \coNP-\complete   & \cite{ptnfas}
                  & \PSpace-\complete & (Thm.~\ref{thmMain})\\
        rpoNFA    &
                  & \NL-\complete& \cite{mfcs16:mktmmt}
                  & \coNP-\complete  & \cite{mfcs16:mktmmt}
                  & \PSpace-\complete & \cite{mfcs16:mktmmt}\\
        poNFA     & $\frac{3}{2}$
                  & \NL-\complete & \cite{mfcs16:mktmmt}
                  & \PSpace-\complete & \cite{mfcs16:mktmmt}
                  & \PSpace-\complete & \cite{AhoHU74} \\
        NFA       &
                  & \coNP-\complete & \cite{StockmeyerM73}
                  & \PSpace-\complete & \cite{AhoHU74} 
                  & \PSpace-\complete & \cite{AhoHU74}
      \end{tabular}
    \end{center}
  \end{table}

  An overview of the results is summarized in Table~\ref{table_results}.

  All proofs or their parts omitted in the paper may be found in the appendix.

\section{Preliminaries}
  We assume that the reader is familiar with automata theory~\cite{AhoHU74}. The cardinality of a set $A$ is denoted by $|A|$ and the power set of $A$ by $2^A$. The empty word is denoted by $\eps$. For a word $w=xyz$, $x$ is a {\em prefix}, $y$ a {\em factor\/}, and $z$ a {\em suffix\/} of $w$. A prefix (factor, suffix) of $w$ is {\em proper\/} if it is different from $w$.

  Let $\A = (Q,\Sigma,\cdot,I,F)$ be a {\em nondeterministic finite automaton\/} (NFA). The language {\em accepted\/} by $\A$ is the set $L(\A) = \{w\in\Sigma^* \mid I \cdot w \cap F \neq \emptyset\}$. We often omit $\cdot$ and write $Iw$ instead of $I\cdot w$. 
  A {\em path\/} $\pi$ from a state $q_0$ to a state $q_n$ under a word $a_1a_2\cdots a_{n}$, for some $n\ge 0$, is a sequence of states and input symbols $q_0 a_1 q_1 a_2 \ldots q_{n-1} a_{n} q_n$ such that $q_{i+1} \in q_i\cdot a_{i+1}$, for all $i=0,1,\ldots,n-1$. Path $\pi$ is {\em accepting\/} if $q_0\in I$ and $q_n\in F$. We write $q_0 \xrightarrow{a_1a_2\cdots a_{n}} q_{n}$ to denote that there is a path from $q_0$ to $q_n$ under the word $a_1a_2\cdots a_{n}$. 
  Automaton $\A$ is {\em complete\/} if for every state $q$ of $\A$ and every letter $a \in \Sigma$, the set $q\cdot a$ is nonempty.
  An NFA $\A$ is {\em deterministic\/} (DFA) if $|I|=1$ and $|q\cdot a|=1$ for every state $q\in Q$ and every letter $a\in \Sigma$. 

  The reachability relation $\le$ on the set of states is defined by $p\le q$ if there exists a word $w\in \Sigma^*$ such that $q\in p\cdot w$. An NFA $\A$ is {\em partially ordered (poNFA)\/}  if the reachability relation $\le$ is a partial order. For two states $p$ and $q$ of $\A$, we write $p < q$ if $p\le q$ and $p\ne q$. A state $p$ is {\em maximal\/} if there is no state $q$ such that $p < q$. Partially ordered automata are sometimes called acyclic automata, where self-loops are allowed.
  
  A \emph{restricted partially ordered NFA (rpoNFA)} is a poNFA that is self-loop deterministic in the sense that the automaton contains no pattern of Figure~\ref{fig_bad_pattern} (right). Formally, for every state $q$ and every letter $a$, if $q\in q\cdot a$ then $q\cdot a = \{q\}$.
  
  A {\em saturated poNFA} ({\em spoNFA}) is a poNFA $\A$ such that, for every state $q$ and every letter $a$, $q\xrightarrow{~a~} q$ is a transition in $\A$.
  \begin{theorem}\label{thmSaturated}
    Universality of spoNFAs is decidable in \ACzero (hence strictly simpler than \LogSpace).
  \end{theorem}
  \begin{proof}
    The language of an spoNFA is universal if and only if $\eps$ belongs to it, which is if and only if one of the initial states is accepting. Given a binary encoding of states and their properties, a family of uniform constant-depth circuits can check if any state has initial and accepting bits set. Unbounded fan-in gates allow us to test any number of (initial) states.
  \end{proof}

\section{Confluent and Self-Loop Deterministic poNFAs -- ptNFAs}
  A poNFA $\A$ over $\Sigma$ with the state set $Q$ can be turned into a directed graph $G(\A)$ with the set of vertices $Q$ where a pair $(p,q) \in Q \times Q$ is an edge in $G(\A)$ if there is a transition from $p$ to $q$ in $\A$. For an alphabet $\Gamma \subseteq \Sigma$, we define the directed graph $G(\A,\Gamma)$ with the set of vertices $Q$ by considering only those transitions corresponding to letters in $\Gamma$. 
  For a state $p$, let $\Sigma(p)=\{a\in\Sigma \mid p\xrightarrow{\,a\,} p\}$ denote all letters labeling self-loops in $p$. 
  We say that $\A$ satisfies the {\em unique maximal state\/} (UMS) property if, for every state $q$ of $\A$, state $q$ is the unique maximal state of the connected component of $G(\A,\Sigma(q))$ containing $q$.
  
  An NFA $\A$ is a {\em ptNFA\/} if it is partially ordered, complete and satisfies the UMS property.

  An equivalent notion to the UMS property for DFAs is confluence~\cite{KlimaP13}. A DFA $\D$ over $\Sigma$ is {\em (locally) confluent\/} if, for every state $q$ of $\D$ and every pair of letters $a, b \in \Sigma$, there is a word $w \in \{a, b\}^*$ such that $(q a) w = (q b) w$. We generalize this notion to NFAs as follows. An NFA $\A$ over $\Sigma$ is {\em confluent\/} if, for every state $q$ of $\A$ and every pair of (not necessarily distinct) letters $a, b \in \Sigma$, if $s \in qa$ and $t\in qb$, then there is a word $w \in \{a, b\}^*$ such that $sw \cap tw \neq\emptyset$. The following relationship between ptNFAs and rpoNFAs holds~\cite{ptnfas}.
  
  \begin{lemma}\label{ptNFAvsrpoNFA}
    Complete and confluent rpoNFAs are exactly ptNFAs.
  \end{lemma}

  As a result, complete and confluent rpoNFAs characterize piecewise testable languages~\cite{ptnfas}, a class of languages that recently re-attracted attention because of some applications in logic on words~\cite{PlaceZ_icalp14} and in XML Schema languages~\cite{icalp2013,HofmanM15,MartensNNS15}; see also Masopust~\cite{ptnfas} for a brief overview of their applications in mathematics and computer science.
  
  We now study the universality problem for ptNFAs. Recall that if the alphabet is fixed, deciding universality for ptNFAs is \coNP-complete and that hardness holds even if restricted to binary alphabets~\cite{ptnfas}. If the alphabet is unary, universality for ptNFAs is decidable in polynomial time~\cite{mfcs16:mktmmt}. We now show that it is \NL-complete.
  
  \begin{theorem}\label{thmMainNL}
    The universality problem for ptNFAs over a unary alphabet is \NL-complete.
  \end{theorem}

  In the case the alphabet may grow polynomially, the universality problem for ptNFAs is open. In the rest of this paper we solve this problem by showing the following result.
  \begin{reptheorem}{thmMain}
    The universality problem for ptNFAs is \PSpace-complete.
  \end{reptheorem}
  
  A typical proof showing \PSpace-hardness of universality for NFAs is to take a $p$-space bounded deterministic Turing machine $\M$, for a polynomial $p$, together with an input $x$, and to encode the computations of $\M$ on $x$ as words over some alphabet $\Sigma$ that depends on the alphabet and the state set of $\M$. One then constructs a regular expression (or an NFA) $R_x$ representing all computations that do not encode an accepting run of $\M$ on $x$. That is, $L(R_x)=\Sigma^*$ if and only if $\M$ does not accept $x$~\cite{AhoHU74,BexGMN09,mfcs16:mktmmt}.
  The form of $R_x$ is relatively simple, consisting of a union of expressions of the form
  \begin{equation}\label{eq_part}
    \Sigma^* \, K \, \Sigma^*
  \end{equation}
  where $K$ is a finite language with words of length bounded by $O(p(|x|))$. 
  Intuitively, $K$ encodes possible violations of a correct computation of $\M$ on $x$, such as the initial configuration does not contain the input $x$, or the step from a configuration to the next one does not correspond to any rule of $\M$. These checks are local, involving at most two consecutive configurations of $\M$, each of polynomial size. They can therefore be encoded as a finite language with words of polynomial length. The initial $\Sigma^*$ of \eqref{eq_part} nondeterministically guesses a position in the word where a violation encoded by $K$ occurs, and the last $\Sigma^*$ reads the rest of the word if the violation check was successful.
  
  This idea cannot be directly used to prove Theorem~\ref{thmMain} for two reasons:
  \begin{description}
    \item[(i)] Although expression~\eqref{eq_part} can easily be translated to a poNFA, it is not true for ptNFAs. The translation of the leading part $\Sigma^* K$ may result in the forbidden pattern of Figure~\ref{fig_bad_pattern};
    \item[(ii)] The constructed poNFA may be incomplete and its ``standard'' completion by adding the missing transitions to a new sink state may violate the UMS property.
  \end{description}

  A first observation to overcome these problems is that the length of the encoding of a computation of $\M$ on $x$ is at most exponential with respect to the size of $\M$ and $x$. It would therefore be sufficient to replace the initial $\Sigma^*$ in \eqref{eq_part} by prefixes of an exponentially long word. However, such a word cannot be constructed by a polynomial-time reduction. Instead, we replace $\Sigma^*$ with a ptNFA encoding such a word, which exists and is of polynomial size as shown in Lemma~\ref{exprponfas}. There we construct, in polynomial time, a ptNFA $\A_{n,n}$ that accepts all words but a single one, $W_{n,n}$, of exponential length. 

  Since language $K$ of \eqref{eq_part} is finite, there is a ptNFA for $K$. For every state of $\A_{n,n}$, we make a copy of the ptNFA for $K$ and identify its initial state with the state of $\A_{n,n}$ if it does not violate the forbidden pattern of Figure~\ref{fig_bad_pattern}; see Figure~\ref{ideaIllustration} for an illustration. We keep track of the words read by both $\A_{n,n}$ and the ptNFA for $K$ by taking the Cartesian product of their alphabets. A letter is then a pair of symbols, where the first symbol is the input for $\A_{n,n}$ and the second is the input for the ptNFA for $K$. A word over this alphabet is accepted if the first components do not form $W_{n,n}$ or the second components form a word that is not a correct encoding of a run of $\M$ on $x$. This results in an rpoNFA that overcomes problem {\sffamily\bfseries(i)}. 
  
  \begin{figure}
    \centering
    \includegraphics[scale=.8]{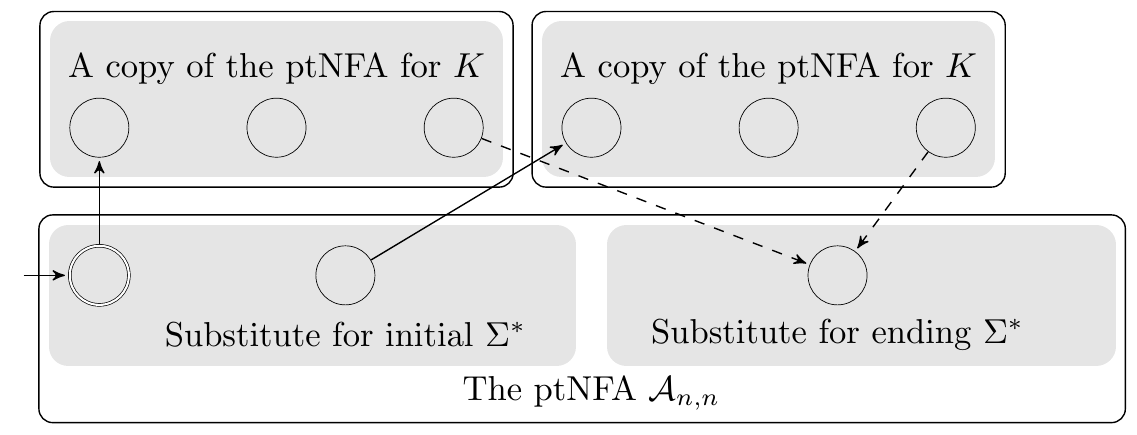}
    \caption{Construction of an rpoNFA (solid edges) solving problem {\sffamily\bfseries(i)}, illustrated for two copies of the ptNFA for $K$, and its completion to a ptNFA (dashed edges) solving problem {\sffamily\bfseries(ii)}}
    \label{ideaIllustration}
  \end{figure}

  However, this technique is not sufficient to resolve problem {\sffamily\bfseries(ii)}. Although the construction yields
  an rpoNFA that is universal if and only if the regular expression $R_x$ is~\cite{mfcs16:mktmmt}, the rpoNFA is incomplete and its ``standard'' completion by adding the missing transitions to an additional sink state violates the UMS property. According to Lemma~\ref{ptNFAvsrpoNFA}, to construct a ptNFA from the rpoNFA, we need to complete the latter so that it is confluent. This is not possible for every rpoNFA, but it is possible for our case because the length of the input that is of interest is bounded by the length of $W_{n,n}$. The maximal state of $\A_{n,n}$ is accepting, therefore all the missing transitions can be added so that the paths required by confluence meet in the maximal state of $\A_{n,n}$. Since all words longer than $|W_{n,n}|$ are accepted by $\A_{n,n}$, we could complete the rpoNFA by adding paths to the maximal state of $\A_{n,n}$ that are longer than $|W_{n,n}|$. However, this cannot be done by a polynomial-time reduction, since the length of $W_{n,n}$ is exponential. Instead, we add a ptNFA to encode such paths in the formal definition of $\A_{n,n}$ as given in Lemma~\ref{exprponfas} below. 
  We then ensure confluence by adding the missing transitions to states of the ptNFA $\A_{n,n}$ from which the unread part of $W_{n,n}$ is not accepted and from which the maximal state of $\A_{n,n}$ is reachable under the symbol of the added transition (cf. Corollary~\ref{WnnStructure}). The second condition ensures confluence, since all the transitions meet in the maximal state of $\A_{n,n}$. The idea is illustrated in Figure~\ref{ideaIllustration}. The details follow.
  
  By this construction, we do not get the same language as defined by the regular expression $R_x$, but the language of the constructed ptNFA is universal if and only if $R_x$ is, which suffices. 

  Thus, the first step is to construct the ptNFA $\A_{n,n}$ that accepts all words but $W_{n,n}$ of exponential length. This automaton is the core of the proof of Theorem~\ref{thmMain} and its construction is described in the following lemma. The language considered there is the same as in our previous work~\cite[Lemma~17]{mfcs16:mktmmt}, where the constructed automaton is not a ptNFA.
  \begin{lemma}\label{exprponfas}
    For all integers $k,n\geq 1$, there exists a ptNFA $\A_{k,n}$ over an $n$-letter alphabet with $n(2k+1)+1$ states, such that the unique non-accepted word of $\A_{k,n}$ is of length $\binom{k+n}{k}-1$.
  \end{lemma}
  \begin{proof}
    For positive integers $k$ and $n$, we recursively define words $W_{k,n}$ over the alphabet $\Sigma_n = \{a_1,a_2,\ldots, a_n\}$ as follows. For the base cases, we set $W_{k,1} = a_1^k$ and $W_{1,n} = a_1a_2\ldots a_n$. The cases for $k,n >1$ are defined recursively by setting
    \begin{align}
      W_{k,n} & = W_{k,n-1}\, a_{n}\, W_{k-1,n}
         = W_{k,n-1}\, a_n\, W_{k-1,n-1}\, a_n\, \cdots\, a_n\, W_{1,n-1}\, a_n\,.
    \end{align}
    The length of $W_{k,n}$ is $\binom{k+n}{n}-1$~\cite{dlt15}. Notice that letter $a_n$ appears exactly $k$ times in $W_{k,n}$. We further set $W_{k,n}=\eps$ whenever $kn=0$, since this is useful for defining $\A_{k,n}$ below. 

    We construct a ptNFA $\A_{k,n}$ over $\Sigma_n$ that accepts the language $\Sigma_n^* \setminus \{W_{k,n}\}$. For $n=1$ and $k\ge 0$, let $\A_{k,1}$ be a DFA for $\{a_1\}^* \setminus \{a_1^k\}$ with $k$ additional unreachable states used to address problem {\sffamily\bfseries(ii)} and included here for uniformity (see Corollary~\ref{WnnStructure}). $\A_{k,1}$ consists of $2k+1$ states of the form $(i;1)$ and a state {\em max}, as shown in the top-most row of states in Figure~\ref{ptnfa3}, together with the given $a_1$-transitions. All states but $(i;1)$, for $i=k,\ldots,2k$, are accepting, and $(0;1)$ is initial. All undefined transitions in Figure~\ref{ptnfa3} go to state {\em max}.

    Given a ptNFA $\A_{k,n-1}$, we recursively construct $\A_{k,n}$ as defined next. The construction for $n=3$ is illustrated in Figure~\ref{ptnfa3}. We obtain $\A_{k,n}$ from $\A_{k,n-1}$ by adding $2k+1$ states $(0;n),(1;n),\ldots,(2k;n)$, where $(0;n)$ is added to the initial states, and all states $(i;n)$ with $i< k$ are added to the accepting states. The automaton $\A_{k,n}$ therefore has $n(2k+1) + 1$ states.
    \begin{figure}
      \centering
        \includegraphics[scale=.46]{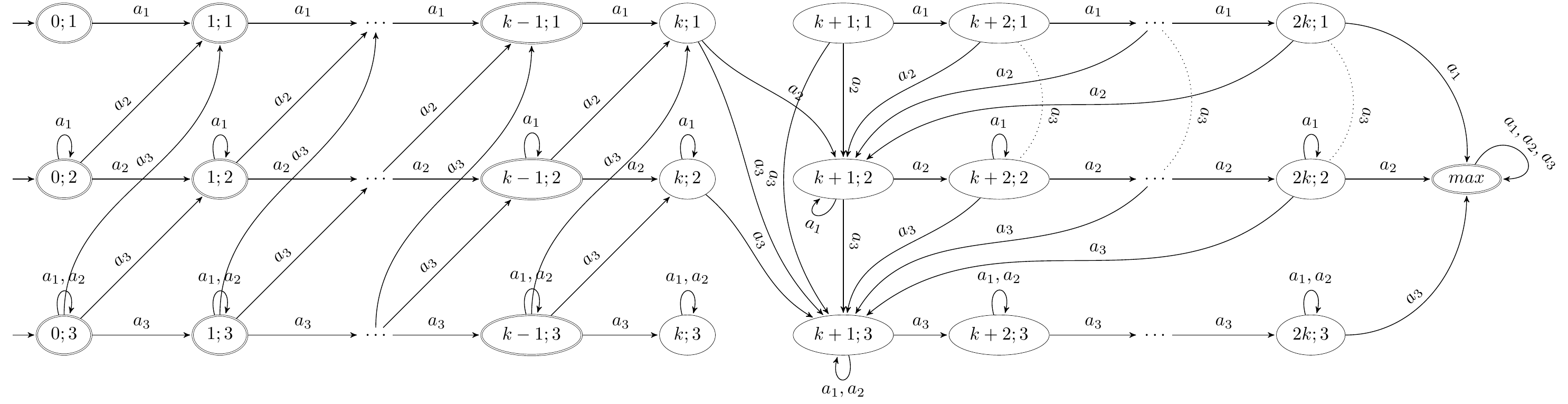}
        \caption{The ptNFA $\A_{k,3}$ with $3(2k+1) + 1$ states; all undefined transitions go to state {\em max}; dotted lines depict arrows from $(k+i,1)$ to $(k+1,3)$ under $a_3$, for $i=2,3,\ldots,k$}
        \label{ptnfa3}
    \end{figure}
    The additional transitions of $\A_{k,n}$ consist of the following groups:
    \begin{enumerate}
      \item\label{r1} Self-loops $(i;n)\xrightarrow{a_j}(i;n)$ for every $i\in\{0,1,\ldots,2k\}$ and $a_j=a_1,a_2,\ldots,a_{n-1}$;
        
      \item\label{r2} Transitions $(i;n)\xrightarrow{a_n}(i+1;n)$ for every $i\in\{0,1,\ldots,2k-1\}\setminus\{k\}$;
        
      \item\label{r3} Transitions $(k,n)\xrightarrow{a_n} max$ and $(2k,n)\xrightarrow{a_n} max$, and the self-loop $max \xrightarrow{a_n} max$;
        
      \item\label{r4} Transitions $(i;n)\xrightarrow{a_n}(i+1;m)$ for every $i=0,1,\ldots,k-1$ and $m=1,\ldots,n-1$;
        
      \item\label{r5} Transitions $(i;m)\xrightarrow{a_n} max$ for every accepting state $(i;m)$ of $\A_{k,n-1}$;
        
      \item\label{r6} Transitions $(i;m)\xrightarrow{a_n} (k+1,n)$ for every non-accepting state $(i;m)$ of $\A_{k,n-1}$.
    \end{enumerate}

    By construction, $\A_{k,n}$ is complete and partially ordered. It satisfies the UMS property because if there is a self-loop in a state $q\neq max$ under a letter $a$, then there is no other incoming or outgoing transition of $q$ under $a$. This means that the component of the graph $G(\A_{k,n},\Sigma(q))$ containing $q$ is only state $q$, which is indeed the unique maximal state. Hence, it is a ptNFA. Equivalently, to see that the automaton is confluent, the reader may notice that the automaton has a single sink state.
      
    We show that $\A_{k,n}$ accepts $\Sigma_n^*\setminus \{W_{k,n}\}$. The additional states of $\A_{k,n}$ and transitions {\sffamily\bfseries\ref{r1}}, {\sffamily\bfseries\ref{r2}}, and {\sffamily\bfseries\ref{r3}} ensure acceptance of every word that does not contain exactly $k$ occurrences of $a_n$.
    The transitions {\sffamily\bfseries\ref{r4}} and {\sffamily\bfseries\ref{r5}} ensure acceptance of all words in $(\Sigma_{n-1}^* a_n)^{i} L(\A_{k-i,n-1})a_n \Sigma_n^*$, for which the longest factor before the $(i+1)$th occurrence of $a_n$ is not of the form $W_{k-i,n-1}$, hence not a correct factor of $W_{k,n} = W_{k,n-1} a_n \cdots a_n W_{k-i,n-1} a_n \cdots a_n W_{1,n-1} a_n$.
    Together, these conditions ensure that $\A_{k,n}$ accepts every input other than $W_{k,n}$.

    It remains to show that $\A_{k,n}$ does not accept $W_{k,n}$, which we do by induction on $(k,n)$. We start with the base cases. For $(0,n)$ and any $n\geq 1$, the word $W_{0,n}=\eps$ is not accepted by $\A_{0,n}$, since the initial states $(0;m)=(k;m)$ of $\A_{0,n}$ are not accepting. Likewise, for $(k,1)$ and any $k\ge 0$, we find that $W_{k,1}=a_1^k$ is not accepted by $\A_{k,1}$ (cf. Figure~\ref{ptnfa3}).

    For the inductive case $(k,n)\ge (1,2)$, assume that $\A_{k',n'}$ does not accept $W_{k',n'}$ for any $(k',n') < (k,n)$. We have $W_{k,n} = W_{k,n-1} a_n W_{k-1,n}$, and $W_{k,n-1}$ is not accepted by $\A_{k,n-1}$ by induction. Therefore, after reading $W_{k,n-1} a_n$, automaton $\A_{k,n}$ must be in one of the states $(1;m)$, $1\le m\le n$, or $(k+1;n)$. However, states $(1;m)$, $1\le m\le n$, are the initial states of $\A_{k-1,n}$, which does not accept $W_{k-1,n}$ by induction. Assume that $\A_{k,n}$ is in state $(k+1;n)$ after reading $W_{k,n-1} a_n$. Since $W_{k-1,n}$ has exactly $k-1$ occurrences of letter $a_n$, $\A_{k,n}$ is in state $(2k;n)$ after reading $W_{k-1,n}$. Hence $W_{k,n}$ is not accepted by $\A_{k,n}$.
  \end{proof}
    
  The last part of the previous proof shows that the suffix $W_{k-1,n}$ of the word $W_{k,n}=W_{k,n-1}a_nW_{k-1,n}$ is not accepted from state $(k+1;n)$. This can be generalized as follows. 
  \begin{corollary}\label{WnnStructure}
    For any suffix $a_i w$ of $W_{k,n}$, $w$ is not accepted from state $(k+1;i)$ of $\A_{k,n}$.
  \end{corollary}

  The proof of Lemma~\ref{exprponfas} also shows that the transitions of {\sffamily\bfseries\ref{r6}} are redundant. We thus have the following observation. 
  \begin{corollary}\label{cor_nonacc}
    Removing from $\A_{k,n}$ the non-accepting states $(k+1,i),\ldots,(2k,i)$, for $1\le i \le n$, and the corresponding transitions results in an rpoNFA that accepts the same language.
  \end{corollary}

  Since $\binom{2n}{n}\geq 2^n$, Lemma~\ref{exprponfas} implies that there are ptNFAs $\A_{n,n}$ for which the shortest non-accepted word $W_{n,n}$ is exponential in the size of $\A$. 
  
  A {\em deterministic Turing machine} (DTM) is a tuple $M = (Q,T,I,\delta,\blank,q_o,q_f)$, where $Q$ is the finite state set, $T$ is the tape alphabet, $I\subseteq T$ is the input alphabet, $\blank\in T \setminus I$ is the blank symbol, $q_o$ is the initial state, $q_f$ is the accepting state, and $\delta$ is the transition function mapping $Q\times T$ to $Q\times T \times \{L,R,S\}$; see Aho et al.~\cite{AhoHU74} for details.

  We now prove the main result, whose proof is a nontrivial generalization of our previous construction showing \PSpace-hardness of universality for rpoNFAs~\cite{mfcs16:mktmmt}.
  \begin{theorem}\label{thmMain}
    The universality problem for ptNFAs is \PSpace-complete.
  \end{theorem}
  \begin{proof}
    Membership follows from the fact that universality is in \PSpace for NFAs~\cite{GareyJ79}. 
    
    To prove \PSpace-hardness, we consider a polynomial $p$ and a $p$-space-bounded DTM $\M = (Q,T,I,\delta,\blank,q_o,q_f)$. Without loss of generality, we assume that $q_o\neq q_f$. 
    A configuration of $\M$ on $x$ consists of a current state $q\in Q$, the position $1\leq \ell\leq p(|x|)$ of the head, and the tape contents $\theta_1,\ldots,\theta_{p(|x|)}$ with $\theta_i\in T$. We represent it by a sequence 
    \[
      \tuple{\theta_1,\eps}\cdots\tuple{\theta_{\ell-1},\eps}\tuple{\theta_{\ell},q}\tuple{\theta_{\ell+1},\eps}\cdots\tuple{\theta_{p(|x|)},\eps}
    \]
    of symbols from $\Delta = T\times(Q\cup\{\eps\})$. A run of $\M$ on $x$ is represented as a word
    $\# w_1 \# w_2 \# \cdots \allowbreak \# w_m \#$, where $w_i\in\Delta^{p(|x|)}$ and
    $\#\notin\Delta$ is a fresh separator symbol.
    One can construct a regular expression recognizing all words over $\Delta\cup\{\#\}$
    that do not correctly encode a run of $\M$ (in particular are not of the form $\# w_1 \# w_2 \# \cdots \allowbreak \# w_m \#$) or that encode a run that is not accepting~\cite{AhoHU74}. Such a regular expression can be constructed in the following three steps:
    \begin{description}
      \item[(A)] we detect all words that do not start with the initial configuration; 
      
      \item[(B)] we detect all words that do not encode a valid run since they violate a transition rule;
      
      \item[(C)] we detect all words that encode non-accepting runs or runs that end prematurely.
    \end{description}
    
    If $\M$ has an accepting run, it has one without repeated configurations. For an input $x$, there are $C(x) = (|T\times(Q\cup\{\eps\})|)^{p(|x|)}$ distinct configuration words in our encoding. Considering a separator symbol $\#$, the length of the encoding of a run without repeated configurations is at most $1+ C(x)(p(|x|)+1)$, since every configuration word ends with $\#$ and is thus of length $p(|x|)+1$. Let $n$ be the least number such that $|W_{n,n}|\geq 1+ C(x)(p(|x|)+1)$, where $W_{n,n}$ is the word constructed in Lemma~\ref{exprponfas}. Since $|W_{n,n}|+1=\binom{2n}{n} \ge 2^n$, it follows that $n$ is smaller than $\lceil\log(1+ C(x)(p(|x|)+1))\rceil$, hence polynomial in the size of $\M$ and $x$.

    Consider the ptNFA $\A_{n,n}$ over the alphabet $\Sigma_n=\{a_1,\ldots,a_n\}$ of Lemma~\ref{exprponfas}, and define the alphabet $\Deltaplus = T\times(Q\cup\{\eps\})\cup\{\#,\$\}$. We consider the alphabet 
    $
      \Pi=\Sigma_n\times\Deltaplus
    $
    where the first letter is an input for $\A_{n,n}$ and the second letter is used for encoding a run as described above. Recall that $\A_{n,n}$ accepts all words different from $W_{n,n}$. Therefore, only those words over $\Pi$ are of our interest, where the first components form the word $W_{n,n}$. Since the length of $W_{n,n}$ may not be a multiple of $p(|x|)+1$, we add $\$$ to fill up any remaining space after the last configuration. 
    
    For a word $w=\tuple{a_{i_1},\delta_1}\cdots \tuple{a_{i_\ell},\delta_\ell}\in\Pi^\ell$, we define $w[1]=a_{i_1}\cdots a_{i_\ell} \in \Sigma_n^\ell$ as the projection of $w$ to the first component and $w[2]=\delta_1\ldots\delta_\ell\in\Deltaplus^\ell$ as the projection to the second component. Conversely, for a word $v\in\Deltaplus^*$, we write $\enc(v)$ to denote the set of all words $w\in\Pi^{|v|}$ with $w[2]=v$. Similarly, for $v\in\Sigma_n^*$, $\enc(v)$ denotes the words $w\in\Pi^{|v|}$ with $w[1]=v$. We extend this notation to sets of words. 
    
    Let $\enc(\A_{n,n})$ denote the automaton $\A_{n,n}$ with each transition $q\xrightarrow{a_i} q'$ replaced by all transitions $q\xrightarrow{\pi} q'$ with $\pi\in\enc(a_i)$. Then $\enc(\A_{n,n})$ accepts the language $\Pi^*\setminus \{\enc(W_{n,n})\}$. We say that a word $w$ encodes an accepting run of $\M$ on $x$ if
    $w[1]=W_{n,n}$ and
    $w[2]$ is of the form $\#w_1\#\cdots\# w_m \# \$^j$ 
    such that there is an $i\in\{1,2,\ldots,m\}$ for which
    $\#w_1\#\cdots\#w_i\#$ encodes an accepting run of $\M$ on $x$, 
    $w_k=w_i$ for all $k\in\{i+1,\ldots,m\}$, and
    $j\leq p(|x|)$.
    That is, we extend the encoding by repeating the accepting configuration until we have less than $p(|x|)+1$ symbols before the end of $|W_{n,n}|$ and fill up the remaining places with $\$$.

    For {\bf (A)}, we want to detect all words that do not start with the word
    \[
      w[2] =\#\tuple{x_1,q_0}\allowbreak\tuple{x_2,\eps}\cdots\allowbreak\tuple{x_{|x|},\eps}\tuple{\blank,\eps}\cdots\tuple{\blank,\eps}\#
    \]
    of length $p(|x|)+2$.
    This happens if 
      (A.1) the word is shorter than $p(|x|)+2$, or 
      (A.2) at position $j$, for $0\leq j\leq p(|x|)+1$, there is a letter from the alphabet $\Deltaplus \setminus\{x_j\}$.
    Let $\bar{E}_j=\Sigma_n \times (\Deltaplus\setminus\{x_j\})$ where $x_j$ is the $j$th symbol on the initial tape of $\M$. We can capture (A.1) and (A.2) in the regular expression
    \begin{equation}\label{eq_re_wrong_start}
      \left(\varepsilon + \Pi + \Pi^2 +\ldots+ \Pi^{p(|x|)+1}\right) 
      + \sum_{0\leq j\leq p(|x|)+1} (\Pi^{j} \cdot \bar{E}_j\cdot\Pi^*)
    \end{equation}
    
    Expression \eqref{eq_re_wrong_start} is polynomial in size. It can be captured by a ptNFA as follows. Each of the first $p(|x|)+2$ expressions defines a finite language and can easily be captured by a ptNFA (by a confluent DFA) of size of the expression. The disjoint union of these ptNFAs then form a single ptNFA recognizing the language $\varepsilon + \Pi + \Pi^2 +\ldots+ \Pi^{p(|x|)+1}$.
    
    To express the language $\Pi^{j} \cdot \bar{E}_j\cdot\Pi^*$ as a ptNFA, we first construct the minimal incomplete DFA recognizing this language (states $0,1,\ldots,j,j+1,max$ in Figure~\ref{fig_const_1}). However, we cannot complete it by simply adding the missing transitions to a new sink state because it results in a DFA with two maximal states, $max$ and the sink state, violating the UMS property. Instead, we use a copy of the ptNFA $\enc(\A_{n,n})$ and add the missing transitions from state $j$ under $\enc(x_j)$ to state $(n+1;i)$ if $\enc(x_j)[1]=a_i$; see Figure~\ref{fig_const_1}. Notice that states $(n+1;i)$ are the states $(k+1;i)$ in Figure~\ref{ptnfa3}.
    \begin{figure}
      \centering
      \includegraphics[scale=.55]{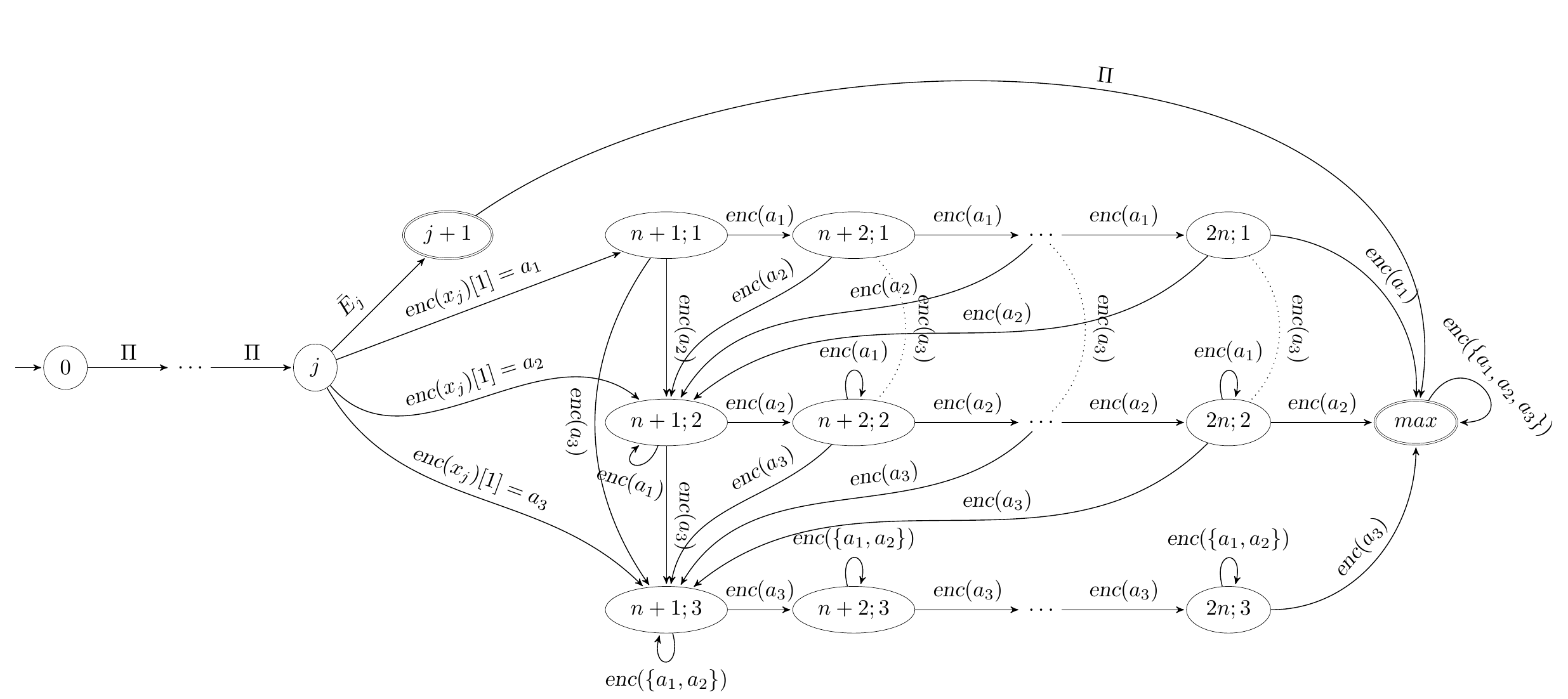}
      \caption{A ptNFA accepting $\Pi^{j} \cdot \bar{E}_j\cdot\Pi^* + (\Pi^*\setminus\{\enc(W_{n,n})\}$ illustrated for $\Sigma_n=\{a_1,a_2,a_3\}$; only the relevant part of $\A_{n,n}$ is depicted}
      \label{fig_const_1}
    \end{figure}
    The resulting automaton is a ptNFA, since it is complete, partially ordered, and satisfies the UMS property -- for every state $q$ different from {\em max}, the component co-reachable and reachable under the letters of self-loops in $q$ is only state $q$ itself. The automaton accepts all words of $\Pi^{j} \cdot \bar{E}_j\cdot\Pi^*$.
    
    We now show that any word $w$ that is accepted by this automaton and that does not belong to $\Pi^{j} \cdot \bar{E}_j\cdot\Pi^*$ is such that $w[1] \neq W_{n,n}$, that is, it belongs to $\Pi^*\setminus\{\enc(W_{n,n})\}$. Assume that $w[1]=W_{n,n}=ua_iv$, where $a_i$ is the position and the letter under which the state $(n+1;i)$ of $\A_{n,n}$ is reached. Then $v$ is not accepted from $(n+1;i)$ by Corollary~\ref{WnnStructure}. Thus, the ptNFA accepts the language $\Pi^{j} \cdot \bar{E}_j\cdot\Pi^* + (\Pi^*\setminus\{\enc(W_{n,n})\})$. Constructing such a ptNFA for polynomially many expressions $\Pi^{j} \cdot \bar{E}_j \cdot \Pi^*$ and taking their union results in a polynomially large ptNFA accepting the language $\sum_{j=0}^{p(|x|)+1} (\Pi^{j} \cdot \bar{E}_j\cdot\Pi^*) + (\Pi^*\setminus \{\enc(W_{n,n})\})$.
    
    Note that we ensure that the surrounding $\#$ in the initial configuration are present.

    For {\bf (B)}, we check for incorrect transitions. Consider again the encoding $\#w_1\#\ldots \allowbreak \#w_m\#$ of a sequence of configurations with a word over $\Delta\cup\{\#\}$. We can assume that $w_1$ encodes the initial configuration according to {\bf (A)}. In an encoding of a valid run, the symbol at any position $j\geq p(|x|)+2$ is uniquely determined by the symbols at positions $j-p(|x|)-2$, $j-p(|x|)-1$, and $j-p(|x|)$, corresponding to the cell and its left and right neighbor in the previous configuration. Given symbols $\delta_\ell,\delta,\delta_r\in\Delta\cup\{\#\}$, we can define $f(\delta_\ell,\delta,\delta_r)\in\Delta\cup\{\#\}$ to be the symbol required in the next configuration. The case where $\delta_\ell=\#$ or $\delta_r=\#$ corresponds to transitions applied at the left and right edge of the tape, respectively; for the case that $\delta=\#$, we define $f(\delta_\ell,\delta,\delta_r)=\#$, ensuring that the separator $\#$ is always present in successor configurations as well. We extend $f$ to $f\colon \Deltaplus^3\to \Deltaplus$. For allowing the last configuration to be repeated, we define $f$ as if the accepting state $q_f$ of $\M$ had a self loop (a transition that does not modify the tape, state, or head position). Moreover, we generally permit $\$$ to occur instead of the expected next configuration symbol. We can then check for invalid transitions using the regular expression
    \begin{equation}\label{eq_re_wrong_trans}
      \Pi^*\, \sum_{\delta_\ell,\delta,\delta_r \in \Deltaplus} \enc(\delta_\ell\delta\delta_r)\cdot  \Pi^{p(|x|)-1} \cdot \hat{f}(\delta_\ell,\delta,\delta_r)\cdot \Pi^*
    \end{equation}
    where $\hat{f}(\delta_\ell,\delta,\delta_r)$ is $\Pi\setminus\enc(\{f(\delta_\ell,\delta,\delta_r),\$\})$. Note that \eqref{eq_re_wrong_trans} detects wrong transitions if a long enough next configuration exists. The case that the run stops prematurely is covered in {\bf (C)}.

    Expression \eqref{eq_re_wrong_trans} is not readily encoded in a ptNFA because of the leading $\Pi^*$. To address this, we replace $\Pi^*$ by the expression $\Pi^{\leq |W_{n,n}|-1}$, which matches every word $w\in\Pi^*$ with $|w|\leq |W_{n,n}|-1$. Clearly, this suffices for our case because the computations of interest are of length $|W_{n,n}|$ and a violation of a correct computation must occur. As $|W_{n,n}|-1$ is exponential, we cannot encode it directly and we use $\enc(\A_{n,n})$ instead.

    In detail, let $E$ be the expression obtained from \eqref{eq_re_wrong_trans} by omitting the initial $\Pi^*$, and let $\B_1$ be an incomplete DFA that accepts the language of $E$ constructed as follows. From the initial state, we construct a tree-shaped DFA corresponding to all words of length three of the finite language $\sum_{\delta_\ell,\delta,\delta_r \in \Deltaplus} \enc(\delta_\ell \delta \delta_r)$. To every leaf state, we add a path under $\Pi$ of length $p(|x|)-1$. The result corresponds to the language $\sum_{\delta_\ell,\delta,\delta_r \in \Deltaplus} \enc(\delta_\ell \delta \delta_r) \cdot \Pi^{p(|x|)-1}$. Let $q_{\delta_\ell \delta \delta_r}$ denote the states uniquely determined by the words in $\enc(\delta_\ell \delta \delta_r) \cdot \Pi^{p(|x|)-1}$. We add the transitions $q_{\delta_\ell \delta \delta_r} \xrightarrow{\enc(\hat{f}(\delta_\ell,\delta,\delta_r))} max'$, where $max'$ is a new accepting state. The automaton is illustrated in the upper part of Figure~\ref{fig_tree}, denoted $\B_1$. It is an incomplete DFA for language $E$ of polynomial size. It is incomplete only in states $q_{\delta_r\delta\delta_\ell}$ due to the missing transitions under $\enc(f(\delta_\ell,\delta,\delta_r))$ and $\enc(\$)$. We complete it by adding the missing transitions to the states of the ptNFA $\A_{n,n}$. Namely, for $z\in \{ \enc(f(\delta_\ell,\delta,\delta_r)), \enc(\$)\}$, we add $q_{\delta_\ell \delta \delta_r} \xrightarrow{~z~} (n+1;i)$ if $z[1]=a_i$. 

    \begin{figure}
      \centering
      \includegraphics[scale=.51]{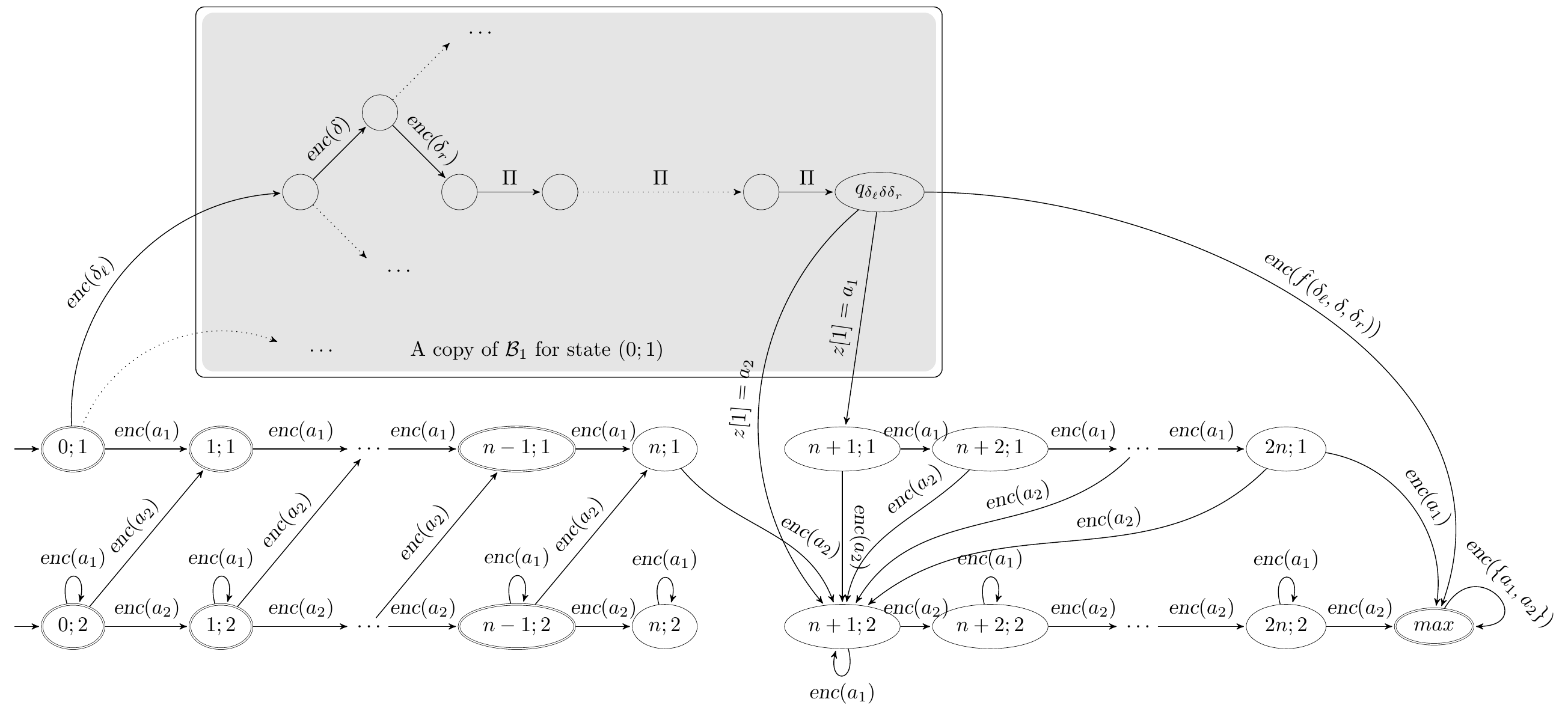}
      \caption{ptNFA $\B$ consisting of $\enc(\A_{n,n})$, $n=2$, with, for illustration, only one copy of ptNFA $\B_1$ for the case the initial state of $\B_1$ is identified with state $(0;1)$ and state $max'$ with state $max$}
      \label{fig_tree}
    \end{figure}
  
    We construct a ptNFA $\B$ accepting the language $(\Pi^*\setminus\{\enc(W_{n,n})\}) + (\Pi^{\leq |W_{n,n}|-1}\cdot E)$ by merging $\enc(\A_{n,n})$ with at most $n(n+1)$ copies of $\B_1$, where we identify the initial state of each such copy with a unique accepting state of $\enc(\A_{n,n})$, if it does not violate the property of ptNFAs (Figure~\ref{fig_bad_pattern}). This is justified by Corollary~\ref{cor_nonacc}, since we do not need to consider connecting $\B_1$ to non-accepting states of $\A_{n,n}$ and it is not possible to connect it to state $max$. We further identify state $max'$ of every copy of $\B_1$ with state $max$ of $\A_{n,n}$. The fact that $\enc(\A_{n,n})$ alone accepts $(\Pi^*\setminus\{\enc(W_{n,n})\})$ was shown in Lemma~\ref{exprponfas}. This also implies that it accepts all words of length $\leq |W_{n,n}|-1$ as needed to show that $(\Pi^{\leq |W_{n,n}|-1}\cdot E)$ is accepted. Entering states of (a copy of) $\B_1$ after accepting a word of length $\geq|W_{n,n}|$ is possible but all such words are longer than $W_{n,n}$ and hence in $(\Pi^*\setminus\{\enc(W_{n,n})\})$.
  
    Let $w$ be a word that is not accepted by (a copy of) $\B_1$. Then, there are words $u$ and $v$ such that $u$ leads $\enc(\A_{n,n})$ to a state from which $w$ is read in a copy of $\B_1$. Since $w$ is not accepted, there is a letter $z$ and a word $v$ such that $uwz$ goes to state $(n+1;i)$ of $\A_{n,n}$ (for $z[1]=a_i$) and $v$ leads $\enc(\A_{n,n})$ from state $(n+1;i)$ to state $max$. If $u[1]w[1]a_iv[1] = W_{n,n,}$, then $v$ is not accepted from $(n+1;i)$ by Corollary~\ref{WnnStructure}, hence $uwzv[1]\neq W_{n,n}$.

    It remains to show that for every proper prefix $w_{n,n}$ of $W_{n,n}$, there is a state in $\A_{n,n}$ reached by $w_{n,n}$ that is the initial state of a copy of $\B_1$, hence the check represented by $E$ in $\Pi^{\leq |W_{n,n}|-1}\cdot E$ can be performed. In other words, if $a_{n,n}$ denotes the letter following $w_{n,n}$ in $W_{n,n}$, then there must be a state reachable by $w_{n,n}$ in $\A_{n,n}$ that does not have a self-loop under $a_{n,n}$. However, this follows from the fact that $\A_{n,n}$ accepts everything but $W_{n,n}$, since then the DFA obtained from $\A_{n,n}$ by the standard subset construction has a path of length $\binom{2n}{n}-1$ labeled with $W_{n,n}$ without any loop. Moreover, any state of this path in the DFA is a subset of states of $\A_{n,n}$, therefore at least one of the states reachable under $w_{n,n}$ in $\A_{n,n}$ does not have a self-loop under $a_{n,n}$.

    The ptNFA $\B$ therefore accepts the language $\Pi^{\leq |W_{n,n}|-1}\cdot E + (\Pi^*\setminus\{\enc(W_{n,n})\})$.

    Finally, for {\bf (C)}, we detect all words that 
      (C.1) end in a configuration that is incomplete (too short),
      (C.2) end in a configuration that is not in the accepting state $q_f$,
      (C.3) end with more than $p(|x|)$ trailing $\$$, or
      (C.4) contain $\$$ not only at the last positions, that is, we detect all words where $\$$ is followed by a different symbol.
    For a word $v$, we use $v^{\leq i}$ to abbreviate $\varepsilon + v + \ldots + v^i$, and we define $\bar{E}_f= (T\times (Q\setminus\{q_f\}))$. 
    \begin{description}
      \item[(C.1)] $\Pi^* \enc(\#) (\Pi + \ldots + \Pi^{p(|x|)}) \enc(\$)^{\leq p(|x|)} +{}$
      \item[(C.2)] $\Pi^* \enc(\bar{E}_f) (\varepsilon + \Pi + \ldots + \Pi^{p(|x|)-1}) \enc(\#) \enc(\$)^{\leq p(|x|)} +{}$
      \hfill \makebox[0pt][r]{%
            \begin{minipage}[b][0pt]{\textwidth}
              \begin{equation}\label{eq_re_wrong_final}\end{equation}
          \end{minipage}}
      \item[(C.3)] $\Pi^* \enc(\$)^{p(|x|)+1} +{}$
      \item[(C.4)] $(\Pi\setminus \enc(\$))^* \enc(\$) \enc(\$)^* (\Pi\setminus \enc(\$)) \Pi^*$
    \end{description}

    As before, we cannot encode the expression directly as a ptNFA, but we can perform a similar construction as the one used for encoding \eqref{eq_re_wrong_trans}. 

    The expressions \eqref{eq_re_wrong_start}--\eqref{eq_re_wrong_final} together then detect all non-accepting or wrongly encoded runs of $\M$. In particular, if we start from the correct initial configuration (\eqref{eq_re_wrong_start} does not match), then for \eqref{eq_re_wrong_trans} not to match, all complete future configurations must have exactly one state and be delimited by encodings of $\#$. Expressing the regular expressions as a single ptNFA of polynomial size, we have thus reduced the word problem of polynomially space-bounded Turing machines to the universality problem for ptNFAs. 
  \end{proof}

\section{Discussion}
  Regular languages as well as recursively enumerable languages possess both deterministic and nondeterministic automata models. It is not typical -- deterministic pushdown automata are strictly less powerful than nondeterministic pushdown automata, and the relationship between deterministic and nondeterministic linearly-bounded Turing machines is a longstanding open problem. Surprisingly, piecewise testable languages as well as $\R$-trivial languages possess such a property -- $\R$-trivial languages are characterized by poDFAs~\cite{BrzozowskiF80} as well as by self-loop deterministic poNFAs~\cite{mfcs16:mktmmt}, and piecewise testable languages by confluent poDFAs~\cite{KlimaP13} as well as by complete, confluent and self-loop deterministic poNFAs~\cite{ptnfas}.

  We also point out that the languages of self-loop deterministic poNFAs (and of their restrictions) are definable by deterministic regular expressions~\cite{mfcs16:mktmmt}. Deterministic regular expressions~\cite{Bruggemann-KleinW98a} are of interest in schema languages for XML data, since the W3C standards require the regular expressions in their specification to be deterministic.

  Whether a language is definable by a poNFA or a type thereof has also been investigated. Bouajjani, Muscholl and Touili~\cite{BMT2001} showed that deciding whether a regular language is an Alphabetical Pattern Constraints (hence recognizable by a poNFA) is \PSpace-complete for NFAs, and \NL-complete for DFAs. The complexity is preserved for self-loop deterministic poNFAs~\cite{mfcs16:mktmmt}, for complete, confluent and self-loop deterministic poNFAs~\cite{mfcs2014ex,ChoH91}, and for saturated poNFAs~\cite{Heam02}. In all cases, \PSpace-hardness is a consequence of a more general result by Hunt~III and Rosenkrantz~\cite{HuntR78}. Although the problem whether there is an equivalent self-loop deterministic poNFA for a given DFA was not discussed in the literature, it can be seen that it reduces to checking whether the minimal DFA is partially order~\cite{mfcs16:mktmmt}, which is an \NL-complete problem.

  A characterization of languages in terms of automata with forbidden patterns can be compared to the results of Gla\ss{}er and Schmitz~\cite{GlasserS08,Schmitz}, who used DFAs with a forbidden pattern to obtain a characterization of level~${3}/{2}$ of the dot-depth hierarchy.

  Other relevant classes of partially ordered automata include partially ordered B\"uchi automata~\cite{KufleitnerL11}, two-way poDFAs~\cite{SchwentickTV01}, and two-way poDFAs with look-around~\cite{LodayaPS10}.

\bibliographystyle{plainurl}
\bibliography{mybib}

\appendix
\section{Proofs}
  In this part, we present proofs omitted in the main body of the paper.

  The following result was reported without proof~\cite{ptnfas}. For the convenience of reviewers, we provide the proof here.
  \begin{replemma}{ptNFAvsrpoNFA}
    Complete and confluent rpoNFAs are exactly ptNFAs.
  \end{replemma}
  \begin{proof}
    First, we show that if $\A$ is a ptNFA, then $\A$ is a complete and confluent rpoNFA. Indeed, the definition of ptNFAs says that $\A$ is complete, partially ordered, and does not contain the pattern of Figure~\ref{fig_bad_pattern}, since the pattern violates the UMS property. Thus, it is a complete rpoNFA. To show that $\A$ is confluent, let $r$ be a state of $\A$, and let $a$ and $b$ be letters of its alphabet ($a=b$ is not excluded) such that $ra \ni s \neq t \in rb$. Let $s'$ and $t'$ be any maximal states reachable from $s$ and $t$ under the alphabet $\{a,b\}$, respectively. By the UMS property of $\A$, there is a path from $t'$ to $s'$ under $\Sigma(s')$ and a path from $s'$ to $t'$ under $\Sigma(t')$. Since $\A$ is partially ordered, $s'=t'$, which shows that $\A$ is confluent.

    On the other hand, assume that $\A$ is a complete and confluent rpoNFA. To show that it is a ptNFA, we show that it satisfies the UMS property. For the sake of contradiction, assume that the UMS property is not satisfied, that is, there is a state $q$ in $\A$ such that the component $G(\A,\Sigma(g))$ of $\A$ containing $q$ and consisting only of transitions labeled with $\Sigma(q)$ has at least two maximal states with respect to $\Sigma(q)$. Let $r$ be the biggest state in $G(\A,\Sigma(g))$ with respect to the partial order on states such that at least two different maximal states, say $s\neq t$, are reachable from $r$ under $\Sigma(q)$.
    Such a state exists by assumption. We have that $r\notin \{s,t\}$; indeed, if $r=s$, then $t \in r\cdot au$, for some $a\in\Sigma(q)$ and $u\in\Sigma(q)^*$. Since $\A$ is an rpoNFA, it does not have any patter of Figure~\ref{fig_bad_pattern}, which means that $a \notin \Sigma(s) \supseteq \Sigma(q)$, a contradiction. Let $s' \in ra$ and $t' \in rb$ be two different states on the path from $r$ to $s$ and $t$, respectively, for some letters $a,b\in\Sigma(q)\setminus\Sigma(r)$. Then $r < \min\{s',t'\}$. Since $\A$ is confluent, there exists $r'$ such that $r'\in s'w \cap t'w$, for some $w \in \{a,b\}^*$. Let $r''$ denote a maximal state that is reachable from $r'$ under $\Sigma(q)$. There are three cases: 
    (i) if $r'' = s$, then $r < t'$ and both $s$ and $t$ are reachable from $t'$ under $\Sigma(q)$, which yields a contradiction with the choice of $r$;
    (ii) $r'' = t$ yields a contradiction with the choice of $r$ as in (i) by replacing $t'$ with $s'$; and
    (iii) $r'' \notin \{s,t\}$ yields also a contradiction with the choice of $r$, since $r < \min\{s',t'\}$ and, e.g., $s$ and $r''$ are two different maximal states with respect to $\Sigma(q)$ reachable from $s'$ under $\Sigma(q)$. 
    Thus, $\A$ satisfies the UMS property, which completes the proof.
  \end{proof}

  \begin{reptheorem}{thmMainNL}
    The universality problem for ptNFAs over a unary alphabet is \NL-complete.
  \end{reptheorem}
  \begin{proof}
    The problem is in \NL even for unary poNFAs~\cite{mfcs16:mktmmt_full}. 
    We prove hardness by reduction from the \NL-complete DAG-reachability problem~\cite{Jones75}. Let $G$ be a directed acyclic graph with $n$ nodes, and let $s$ and $t$ be two nodes of $G$. We define a ptNFA $\A$ as follows. With each node of $G$, we associate a state in $\A$. Whenever there is an edge from $i$ to $j$ in $G$, we add a transition $i\xrightarrow{a} j$ to $\A$. In addition, we add $n-1$ new non-accepting states $f_1,\ldots,f_{n-1}$ together with the transitions $f_i \xrightarrow{a} f_{i+1}$, for $i=1,\ldots,n-2$. For every state $q\notin\{t,f_1,\ldots,f_{n-1}\}$, we add a transition $q\xrightarrow{a} f_1$. Finally, we add a self-loop $t\xrightarrow{a} t$ and a transition $f_{n-1} \xrightarrow{a} t$. The initial state of $\A$ is $s$ and all states corresponding to nodes of $G$ are final. The automaton is partially ordered, complete and satisfies the UMS property, since state $t$ is the only state with a self-loop and every path under $a^*$ ends up in it. 
    
    It remains to show that $\A$ is universal if and only if there is a path from $s$ to $t$ in $G$. If $t$ is reachable from $s$ in $G$, then $L(\A)=\Sigma^*$, since $t$ is reachable from $s$ via states corresponding to nodes of $G$, which are all accepting in $\A$. If $t$ is not reachable from $s$ in $G$, then $t$ is reachable from $s$ in $\A$ via the path $s \xrightarrow{a^k} q \xrightarrow{a} f_1\xrightarrow{a} f_2\xrightarrow{a} \ldots \xrightarrow{a} f_{n-1} \xrightarrow{a} t$, for any $q$ corresponding to a node of $V\setminus\{t\}$ reachable from $s$ in $G$. We show that $a^{n-1}$ does not belong to $L(\A)$. The shortest path from state $s$ to state $t$ in $\A$ is of length $n$ for $q=s$. Thus, any word accepted in $t$ is of length at least $n$. On the other hand, every word accepted in a state corresponding to a node of $V\setminus\{t\}$ is of length at most $n-2$, since $|V\setminus\{t\}|=n-1$ and $\A$ is acyclic (without self-loops) on those states. This gives that $a^{n-1}$ is not accepted by $\A$, hence $L(\A)$ is not universal.
  \end{proof}

  \begin{repcorollary}{WnnStructure}
    For any suffix $a_i w$ of $W_{k,n}$, $w$ is not accepted from state $(k+1;i)$ of $\A_{k,n}$.
  \end{repcorollary}
  \begin{proof}
    Consider the word $W_{k,n}$ over $\Sigma_n=\{a_1,a_2,\ldots,a_n\}$ constructed in the proof of Lemma~\ref{exprponfas}, and let $i\in\{1,\ldots,n\}$ be the maximal number for which there is a suffix $a_i w$ of $W_{k,n}$ such that $w$ is accepted by $\A_{k,n}$ from state $(k+1;i)$. Then $W_{k,n}=w_1a_i w_2 w_3$, where $w_2\in\{a_1,\ldots,a_i\}^*$ is the shortest word labeling the path from state $(k+1;i)$ to state $max$. By the construction of $\A_{k,n}$, word $a_iw_2$ must contain $k+1$ letters $a_i$. We shown that $W_{k,n}$ does not contain more than $k$ letters $a_i$ interleaved only with letters $a_j$ for $j<i$, which yields a contradiction that proves the claim.
      
    By definition, every longest factor of $W_{k,n}$ over $\{a_1,\ldots,a_i\}$ is of the form $W_{k-\ell,i}$, for $\ell \in\{0,\ldots,k-1\}$. Since $W_{k-\ell,i} = W_{k-\ell,i-1}\, a_i\, W_{k-\ell-1,i-1}\, a_i\, \cdots\, a_i\, W_{1,i-1}\, a_i$, the number of occurrences of $a_i$ interleaved only with letters $a_j$ for $j<i$ is at most $k-\ell$, which results in $k$ for $\ell=0$ as claimed above.
  \end{proof}

  Here we present the omitted part of the proof of the main theorem.
  \begin{reptheorem}{thmMainNL}
    The universality problem for ptNFAs over a unary alphabet is \NL-complete.
  \end{reptheorem}
  \begin{proof}[Proof (omitted parts)]
        Finally, for {\bf (C)}, we detect all words that 
      (C.1) end in a configuration that is incomplete (too short),
      (C.2) end in a configuration that is not in the accepting state $q_f$,
      (C.3) end with more than $p(|x|)$ trailing $\$$, or
      (C.4) contain $\$$ not only at the last positions, that is, we detect all words where $\$$ is followed by a different symbol.
    For a word $v$, we use $v^{\leq i}$ to abbreviate $\varepsilon + v + \ldots + v^i$, and we define $\bar{E}_f= (T\times (Q\setminus\{q_f\}))$. 
    \begin{align*}
      {\bf (C.1)\qquad} & \Pi^* \enc(\#) (\Pi + \ldots + \Pi^{p(|x|)}) \enc(\$)^{\leq p(|x|)}  +{}\nonumber\\
      {\bf (C.2)\qquad} & \Pi^* \enc(\bar{E}_f) (\varepsilon + \Pi + \ldots + \Pi^{p(|x|)-1}) \enc(\#) \enc(\$)^{\leq p(|x|)} +{} \tag{\ref{eq_re_wrong_final}}\\
      {\bf (C.3)\qquad} & \Pi^* \enc(\$)^{p(|x|)+1} +{} \nonumber \\
      {\bf (C.4)\qquad} & (\Pi\setminus \enc(\$))^* \enc(\$) \enc(\$)^* (\Pi\setminus \enc(\$)) \Pi^* \nonumber
    \end{align*}
    
    As before, we cannot encode the expression directly as a ptNFA, but we can perform a similar construction as in \eqref{eq_re_wrong_trans}. Namely, 
      a ptNFA for C.1 is illustrated in Figure~\ref{figC1}, 
      for C.2 in Figure~\ref{figC2}, and
      for C.3 in Figure~\ref{figC3}. 
      Finally, C.4 can be represented by a three-state partially ordered and confluent DFA.
  \end{proof}

    \begin{figure}[!ht]
      \centering
      \includegraphics[scale=.5]{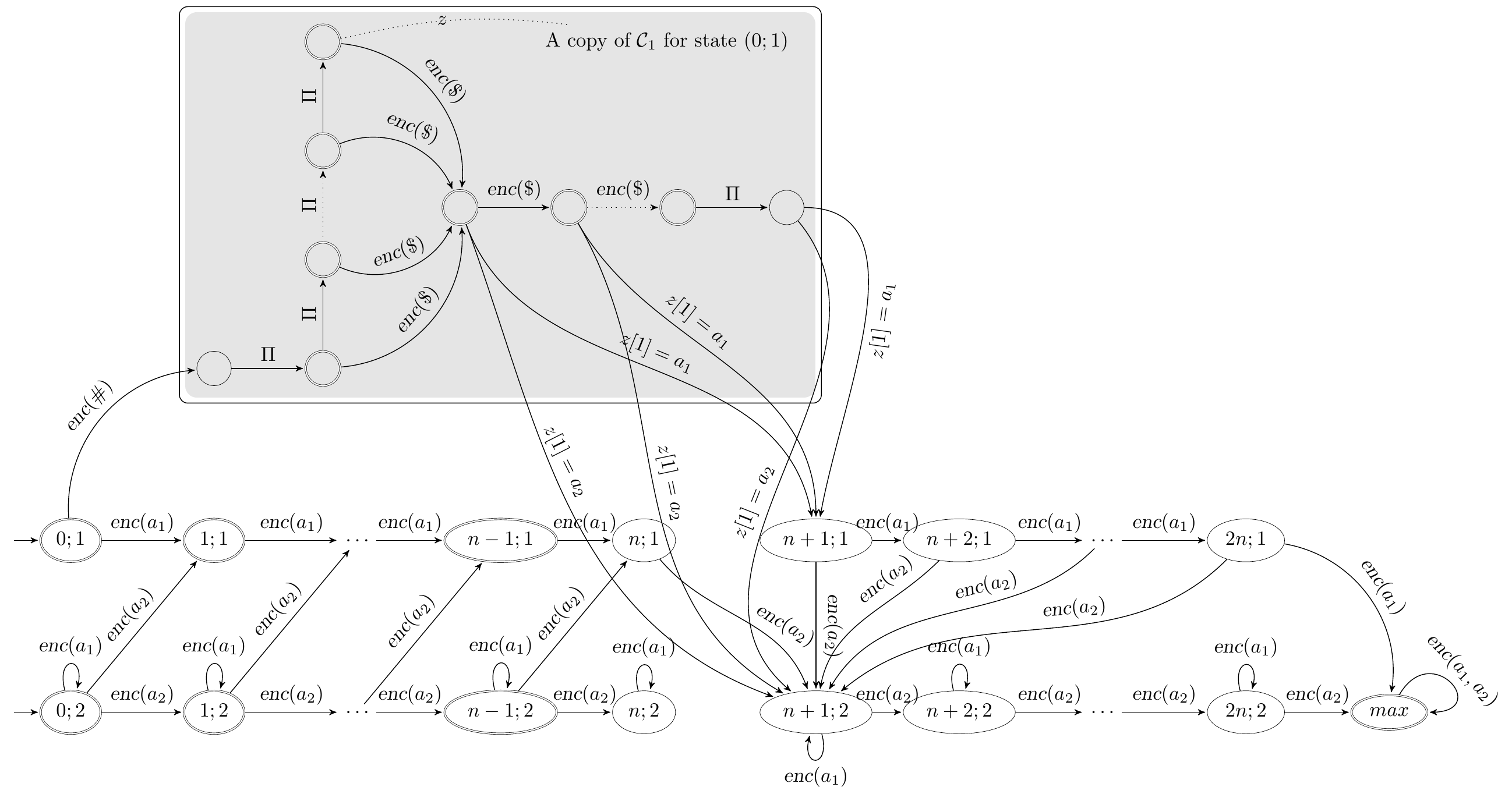}
      \caption{The ptNFA for expression C.1 illustrated for $n=2$}
      \label{figC1}
    \end{figure}

    \begin{figure}[!ht]
      \centering
      \includegraphics[scale=.5]{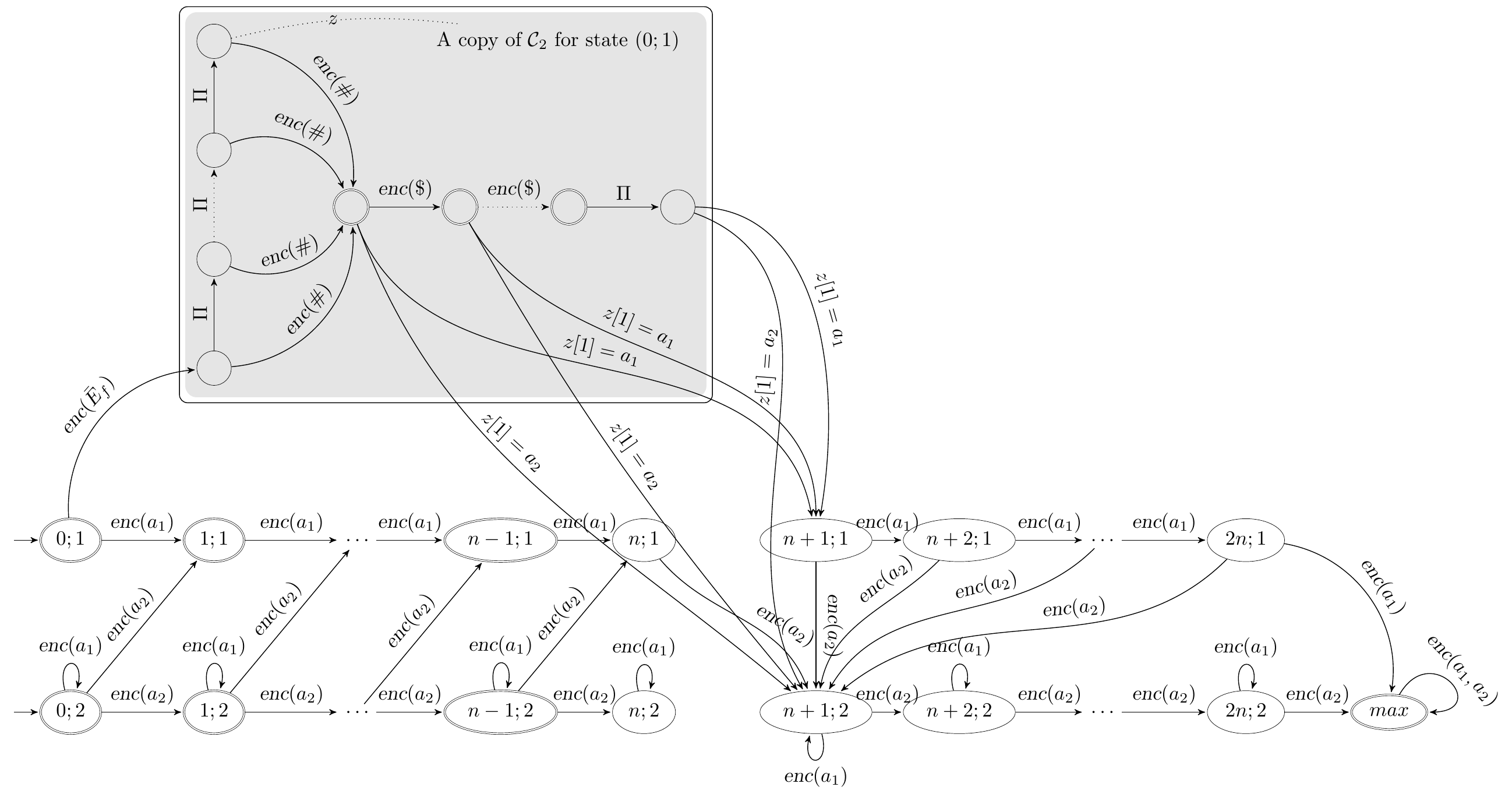}
      \caption{The ptNFA for expression C.2 illustrated for $n=2$}
      \label{figC2}
    \end{figure}

    \begin{figure}[!ht]
      \centering
      \includegraphics[scale=.5]{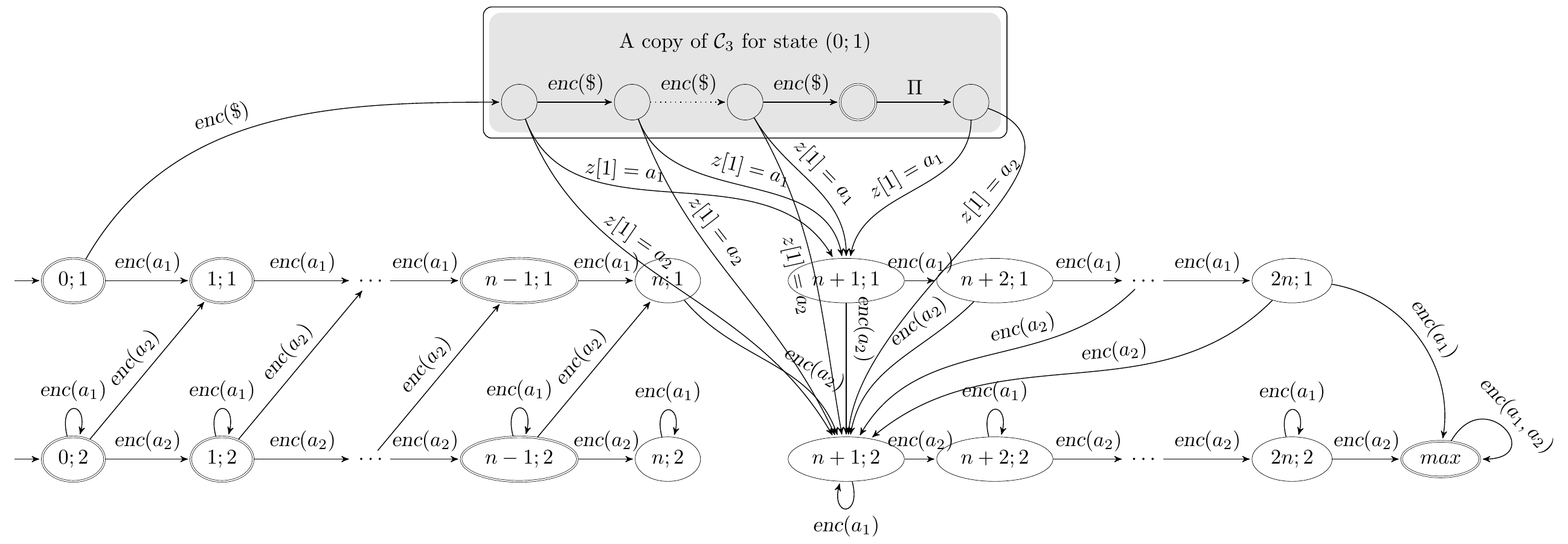}
      \caption{The ptNFA for expression C.3 illustrated for $n=2$}
      \label{figC3}
    \end{figure}

\end{document}